\documentclass[letterpaper,runningheads]{llncs}
\usepackage[utf8]{inputenc}
\usepackage{todonotes}

\usepackage{setspace}
\usepackage{datetime}
\usepackage{tikz}
\usepackage{hyperref}

\usepackage{hyperref}
\usepackage{amssymb}
\usepackage{mathtools}
\usepackage{algorithm}%http://ctan.org/pkg/algorithms
\usepackage{algpseudocode}
\usepackage{braket}
\usepackage{graphicx}
\usepackage{xcolor}
\usepackage{xspace}
\usepackage{hypbmsec}

\newtheorem{result}{Result}

\usepackage{algpseudocode}
\usepackage{algorithm}

\DeclareMathOperator{\ennumerate}{Enumerate}

\DeclareMathOperator{\emptyword}{\varepsilon}

\DeclareMathOperator{\PMN}{WalkMaxNode}

\title{Enumerating $m$-Length Walks in Directed Graphs with Constant Delay}

\title{Enumerating $m$-Length Walks in Directed Graphs with Constant Delay}

\author{Duncan Adamson \inst{1} \and Pawe\l{}  Gawrychowski\inst{2} \and Florin Manea \inst{3}}

\institute{ Materials Innovation Factory, University of Liverpool, UK \and Institute of Computer Science, University of Wrocław, Poland \and Department of Computer Science, University of Göttingen, Göttingen, Germany}

\authorrunning{Adamson, Gawrychowski, Manea}

\begin{document}

\maketitle

\begin{abstract}
    In this paper, we provide a novel enumeration algorithm for the set of all walks of a given length within a directed graph.
    Our algorithm has worst-case constant delay between outputting succinct representations of such walks, after a preprocessing step requiring linear time relative to the size of the graph.
    We apply these results to the problem of enumerating succinct representations of the strings of a given length from a prefix-closed regular language (languages accepted by a finite automaton which has final states only).\looseness=-1
\end{abstract}

\section{Introduction}
% {\colour{red} to update}

Enumerating all members of a given class of combinatorial objects is one of the fundamental problems in computer science. Enumeration problems take a description of the class of objects and produce every object satisfying this description. Often, the number of objects in each class is of exponential size relative to the size of the description: for example, the set of walks of length $m$ in the complete graph $K_n$ with $n$ vertices has size $n^m$. Due to the large size of these classes, the usual goal of enumeration algorithms is to reduce the delay between outputting consecutive objects, either in terms of the worst case or the average case. Various enumeration problems appear in diverse contexts with a wide range of applications. Comprehensive surveys of enumeration problems and their connections to various areas of computer science and mathematics have been provided by Segoufin \cite{segoufin2013enumerating} (with a focus on logic), Wasa \cite{wasa2016enumeration} (which provides a list of enumeration problems from multiple areas, ranging from graph theory to computational geometry or to combinatorics on words and automata), and Uno~\cite{Uno16} (focused on the amortized analysis of enumeration algorithms). Interesting applications of enumeration algorithms include, among others, database theory \cite{amarilli2021constant,amarilli2022efficient,SchmidS21,SchmidS22}, combinatorics and algorithms on strings and the study of formal languages \cite{ackerman2009three,ackerman2009efficient,amarilli2022efficient,Gruber0S21,Shallit13}.\looseness=-1
%ConteGPU22 , or bioinformatics \cite{AngelovHKKK06,gruner1996analysis}.

This paper is primarily motivated by the problem of enumerating the set of all \emph{crystal structures} of a given size.
% which do not contain forbidden factors from a given set.
This problem originates in chemistry, with the problem of \emph{crystal structure prediction}. In one dimension, the crystal structure prediction problem asks, given an alphabet of ``blocks'' (3-dimensional structures), what is the optimal way to arrange these objects to minimise some pairwise objective function~\cite{collins2017accelerated}. Currently, this problem is solved via heuristic techniques \cite{Oganov2018}, leaving the possibility of missing the optimal solution in numerous instances, or by costly integer programming-based algorithms \cite{gusev2023optimality} that often correspond to checking every possible solution. At the same time, existing knowledge from chemistry allows certain solutions to be ruled out without the costly process of simulating the predicted structure~\cite{Oganov2018}, simply based on the presence of some bad (or forbidden) block combinations. By considering the set of blocks as an alphabet, we can canonically represent the set of various arrangements of blocks as strings over this alphabet. In this way, the set of valid arrangements of blocks forms as a \emph{prefix-closed} regular language, i.e., a set of strings where every prefix of every string is also in the set. The class of prefix-closed regular languages has nice language theoretic properties \cite{CevorovaJMPS14,KaoRS09}, and, interestingly w.r.t. our motivation, includes the class of languages of strings that avoid a given set of forbidden factors (corresponding to the bad combinations of blocks). Thus, in this framework, we may solve the problem of enumerating the set of valid crystal structures, for a given set of blocks, by solving the problem of enumerating all strings in a given prefix-closed regular language.

Rather than restricting ourselves purely to prefix-closed regular languages (or crystal structures), and since deterministic finite automata (which are natural ways to specify such languages) can be canonically represented as directed, labelled multi-graphs, we look at the problem of enumerating (succinct representations of) all walks of a given length in a directed graph (potentially parameterized by their starting vertex or set of starting vertices). This seems to lift the language-enumeration problem discussed above to a more abstract setting. Therefore, we will first consider Problem \ref{prob:enum}, and then see how we can use the obtained results to solve the more concrete problem introduced above. \looseness=-1

\begin{problem}
    \label{prob:enum}
    Given directed graph $G$, and integer $m>0$, enumerate efficiently succinct representations of all walks of length $m$ in $G$.\looseness=-1
\end{problem}

To fully specify this problem, one needs to define exactly what is the output of the enumeration. One possibility would be to output each walk of length $m$ explicitly (as a sequence of edges); this would inherently lead to $O(m)$-time delay between the consecutive walks (as we first need to finish outputting the first walk, before starting the next one). To achieve $O(1)$-delay, an implicit, succinct representation of the output is needed. However, such a representation has to be meaningful: one should be able to canonically and efficiently retrieve the explicit list of enumerated walks from the list of implicitly represented walks. Ideally, an algorithm solving Problem \ref{prob:enum} would also permit explicitly outputting, on demand, the current walk of the enumeration, at any step of the computation. Also, ideally, the preprocessing done by the algorithm would take linear time in the size of $G$ and would not depend on $m$, allowing the constructed data structures to be reused to enumerate walks of other lengths. Worth noting, straightforward implicit representations of the enumerated walks (e.g., outputting in step $i$ of the enumeration the implicit description ``the $i^{th}$ walk of $G$ in lexicographic order (as induced by a total order on edges/vertices)'') are usually not meaningful (or do not fulfil our other requirements), as obtaining the walks explicitly would usually require non-trivial work (in the enumeration or in the preprocessing phase).\looseness=-1

\paragraph{Our Contributions.} 
We solve Problem \ref{prob:enum} by providing an enumeration algorithm with worst-case constant delay after a preprocessing phase running in linear time in the total size of the graph $G$, while being independent of $m$. \looseness=-1

Our {\em first contribution}, essential in achieving this result, is to introduce the notion of {\em default walks} in the graph $G$. More precisely, each vertex of $G$ is associated with a \emph{default edge}, corresponding to the first edge on the longest (possibly infinite) walk in the graph starting at that vertex. These \emph{default edges} allow the definition of \emph{default walks} in the graph (walks consisting only of default edges). In this framework, arbitrary walks in the graph can be represented implicitly as the concatenation of multiple default walks (specified by starting vertex and length) and the non-default edges which connect these walks.\looseness=-1

Our {\em second contribution}, is an enumeration algorithm for the walks of length $m$ within a graph $G$, which crucially uses the notions of default edges and walks both for the sake of efficiency of enumeration and as the basis for the output of the enumerated walks. The main idea behind this algorithm is to maintain internally, while going through the walks we want to enumerate, the implicit representation of the current walk as a list of default walks in the graph. By augmenting the graph of default edges with a series of non-trivial data structures, the implicit representation of the enumerated walks allows us to efficiently compute the representation of the next walk in the enumeration and present it in a succinct (yet, meaningful, in the sense mentioned above) way, by simply referencing (by length) the ``prefix'' shared by the current walk and the next one, and then extending it by appending a non-default edge and a single, succinctly represented default walk. This implicit representation also allows us to output explicitly, on-demand, at every point of the computation, the current walk in linear time w.r.t. its length. Finally, we show that only constant time is needed in our enumeration algorithms to move between consecutive walks and output their representations. Our main results are, as such, efficient solutions to Problem \ref{prob:enum}. \looseness=-1
\begin{result}[Theorem \ref{thm:any_starting_vertex}] % and \ref{thm:complexity_var_length}]
\label{res:enum} %(Theorems \ref{thm:complexity_fixed_length} and \ref{thm:complexity_var_length})
    Given directed graph $G = (V, E)$, with $n$ vertices, and integer $m$, we can enumerate succinct representations of all walks of length $m$ in $G$, with $O(1)$-delay, after $O(\vert E \vert)$-time preprocessing.
\end{result}

The above result can be extended to enumerate with $O(1)$-delay all the walks in a graph, whose length is between two given integers $\ell$ and $m$. More interestingly, these results can be directly applied to languages accepted by a prefix closed automaton (PCA for short, a deterministic finite automaton with final states only), by noting that there is a bijective correspondence between the walks starting in the initial state of a PCA and the strings of the language.

\begin{result}[Theorem \ref{thm:enum_PCA}]
    Given integer $m$ and PCA $\mathcal{A}$, we can enumerate succinct representations of all strings of length $m$ accepted by ${\mathcal{A}}$, with $O(1)$ delay after a linear time preprocessing w.r.t. the size of ${\mathcal{A}}$.
\end{result}

% This result can be immediately applied for languages of strings avoiding a given set of forbidden factors ${\mathcal F}$, and the preprocessing time remains $O(m\sigma)$ for $m$ being the sum of the lengths of the factors in ${\mathcal F}$. An interesting byproduct of our approach is that, in the case when ${\mathcal F}$ consists of a single string $f$, the preprocessing time can be lowered to $O(\vert f\vert )$ (Theorem \ref{thm:one_word}). \looseness=-1

% Our first two contributions show that a significant class of regular languages can be enumerated with constant delay after a preprocessing phase taking linear time in the total size of the given automaton. 

We also extend our results to obtain efficient algorithms for the problems of ranking and unranking strings in prefix closed regular languages w.r.t. the order in which they are output by the enumeration algorithm from Theorem \ref{thm:enum_PCA}. 

\begin{result}[Theorem \ref{thm:ranking}]
    \label{res:rank} %(Theorems \ref{thm:ranking} and \ref{thm:unrank})
    Given PCA $A$, with $n$ states, accepting the prefix-closed regular language $L(A)$ over the alphabet $\Sigma$ with $\sigma$ letters, and string $w \in L(A)$, with $\vert w \vert = m$, we can compute in polynomial time the number of strings of length $m$ output in our enumeration of $L(A)$ before outputting $w$. % $O(n m^{\omega} + n(n + \sigma))$
    Moreover, given integer $i$, we can compute in polynomial time (relative to the output) the $i^{th}$ string of length $m$ output in our enumeration algorithm of $L(A)$. % $O(n m^{\omega} + n(n + \sigma))$
\end{result}

\paragraph{Related Work.}
The problem of enumerating walks and paths (i.e., walks with no repeated vertex) within graphs is highly studied with a wealth of existing results. Wasa lists a series of such enumeration tasks and the complexity of their solutions in \cite{wasa2016enumeration}. A notable work regarding the enumeration of paths within a graph is the backtracking technique by Read and Tarjan \cite{read1975bounds}, proving a delay between outputting paths of at most $O(\vert E \vert )$. We note that in the general case, where each path must be explicitly output and there is no upper bound on the length of the path, this is optimal. However, when the length of the paths is bounded by some $m < \vert E \vert$, the potentially significant cost of backtracking limits the efficiency of this algorithm. 
Other initial work on enumerating paths focused on matrix-based approaches, requiring exponential time and space for precomputation, without intermediary output. Danielson \cite{danielson1968finding} (later strengthened by Rubin \cite{Rubin1978Enumeration}) provided such an approach for enumerating all simple paths (paths that do not visit any vertex more than once). Kamae \cite{kamae1967systematic} used this approach to output all cycles and paths in a directed graph; their approach was strengthened by Mateti and Deo \cite{Mateti1976Circuits}, Wild \cite{wild2008generating}, and Birmel{\'e} et al. \cite{birmele2013optimal}. We note that these algorithms are not focused, as ours are, on outputting (succinct representations of) each path (or cycle) sequentially with minimal delay, but rather on collecting all the paths, and thus optimise the time taken to compute the paths efficiently, rather than focusing on the output. Additional work has focused on specific classes of paths, such as Hamiltonian paths \cite{yau1967generation}, ST-paths \cite{birmele2013optimal,Grossi2018EnumeratingSTPaths}, and chordless paths \cite{Uno2014Chordless}.\looseness=-1

Regarding the application of our results on walk-enumeration to the enumeration of strings from prefix closed languages, we also recall the rich literature regarding the enumeration of strings. We again point to the survey of Wasa for a series of classical results \cite{wasa2016enumeration}, as well as to the surveys by Gruber et al. \cite{Gruber0S21} and Shallit \cite{Shallit13}. Several works have focused on more general classes of languages, at the cost of allowing enumeration with non-constant delay, relative to the length of the strings~\cite{ackerman2009three,ackerman2009efficient,SchmidS21,SchmidS22}. Of particular interest to us is the recent work by Amarilli and Monet~\cite{amarilli2022enumerating}, who have provided an algorithm for outputting all strings recognised by a given regular language with bounded delay (which depends on the size of the automaton). This is achieved by partitioning the input language into  \emph{orderable regular languages}, i.e., a language whose strings can be ordered (in a potentially infinite sequence) in such a way that the edit distance between the $i^{th}$ and the $i + 1^{th}$ members of the sequence is bounded. In the other direction, there have been several algorithms for enumerating strings within specific subclasses of prefix closed languages. Quite close to our work are the results of Ruskey and Sawada from \cite{RuskeyS00}, where they give a relatively straightforward algorithm for enumerating the strings of length $n$, over an alphabet with $\sigma$ letters, which do not contain a given factor $f$ of length $m$ with constant amortized delay, after a preprocessing taking $O(m\sigma)$ time. We note that our solution outperforms the one from~\cite{RuskeyS00} for this specific class by reducing the preprocessing time to $O(m)$, and the delay to constant in the worst-case. Similarly, constant amortised delay enumeration algorithms have since been provided for a large number of classes of cyclic strings, %including bracelets \cite{Sawada2001}, fixed-content necklaces \cite{Sawada2003} (cyclic words where the number of occurrences of every symbol over an arbitrary alphabet is fixed), 
most relevantly necklaces and bracelets with a forbidden factor \cite{RuskeyS00,Sawada2001}.\looseness=-1

\section{Preliminaries: Definitions and Sketch of the Algorithm}

% {\colour{red} to update}

The computational model we use in this paper is the RAM with logarithmic word size relative to the size of the input graph or automaton (see below for full details). 

Let $\mathbb{N} = \{1, 2, \ldots\}$ be the set of strictly positive integers and let $[n] = \{1, \ldots, n\}$ for $n \in \mathbb{N}$. Let $G=(V,U)$ be a directed graph (multi-graph) with the set of vertices $V$ and the set of edges $U \subseteq V \times V$ (respectively, $U$ is a multiset of pairs from $V\times V$); the direction of an edge $(v,u)$ is from $v$ to $u$. The directed (multi-)graph $G$ is labelled, with labels over an alphabet $\Sigma = [\sigma]$, if there exists a function ${\mathcal{L}}: U \rightarrow \Sigma$ which labels each edge of $G$ with a letter from $\Sigma$. An example relevant to our paper is that of finite automata, which are directed multi-graphs whose edges, called transitions in that context, are labelled by letters from an input alphabet. We assume that the sets $\Sigma$ and $V$ are totally ordered. \looseness=-1

A walk of length $k$ in $G$ is a sequence $\pi=((v_1,v_2),(v_2,v_3),\ldots,(v_{k-1},v_k))$ such that $(v_i,v_{i+1})\in U$, for all $i\in [k-1]$; the length $k$ of $\pi$ is denoted by $|\pi |$, and $v_1,\ldots,v_k$ are the vertices on the walk $\pi$. Given a walk $\pi $, we refer to its first (respectively, last) $\ell$ edges as the \emph{prefix} (respectively, \emph{suffix}) of length $\ell$ of $\pi$.\looseness=-1

In this paper, we first develop an algorithm for the enumeration of walks of length $m$ within a directed graph $G$, and then apply these results to enumerate strings accepted by a specific class of deterministic finite automaton, called prefix closed automata (PCA, for short). For space reasons, we refer the reader to \cite{HopcroftUllman} for definitions regarding strings and automata.
% are given in Appendix \ref{app:defs}. The reader is also referred to \cite{HopcroftUllman}.
\looseness=-1

\subsection{Computational Model}\label{app:model}

In the problems we consider, the input consists in a natural number $m$, a directed graph with $n$ vertices, and $e = \vert E \vert$ edges. We assume that $e \geq n$ and let $N=\max\{e,m\}$.
% over an ordered alphabet $\{1,\ldots,\sigma\}$ with $\sigma$ letters; let $N=\max\{n,m,\sigma\}$.
The computational model we use is the RAM with logarithmic word size. More precisely, we assume that each memory word can hold ${\mathtt w}\in \Theta( \log N)$ bits and arithmetic operations with numbers in $[N]$ take $O(1)$ time (in particular, we assume that working with numbers upper bounded by $m$ can be done in $O(1)$ time - this includes both arithmetic and input/output operations). Numbers larger than $N$, with $\ell $ bits, are represented in $O(\ell/ \mathtt{w})$ memory words, and working with them takes time proportional to the number of memory words on which they are represented. In the automata-related applications of our graph algorithm, the strings processed in our algorithms are seen as sequences of integers, each fitting in one memory word.

The model described above is common in enumeration algorithms, see~\cite{wasa2016enumeration} and the references therein, especially those addressing problems where the input is a number $m$ and enumerating, generating, ranking, or unranking structures (such as strings, graphs, trees) of size $m$ fulfilling certain properties is required. Similarly, this model is usual in string algorithms, see~\cite{crochemore}.

We can assume that in our computational model, as in many modern programming languages, tail calls to a recursive function are implemented without adding a new stack frame to the call stack. Note that this assumption, as well as the usage of recursivity in our algorithm, is simply aimed to make the presentation lighter. In fact, the algorithms we describe here can be implemented iteratively (although in a more tedious manner, which makes use of the stacks maintained by that algorithm to simulate the recursive calls), and, in that setting, there is no need to manage tail recursion; their complexity remains unchanged. 

\subsection{Strings and Automata: Definitions} \label{app:defs}

Let $\Sigma =[\sigma]$ be a totally, strictly ordered alphabet, with $\sigma$ letters, namely $1 < 2 < \ldots < \sigma$. By $\Sigma^+$ we denote the set of non-empty strings (words) over $\Sigma$ and $\Sigma^* = \Sigma^+ \cup \{\emptyword\}$ (where $\emptyword$ is the empty string). For a string $w \in \Sigma^*$, we denote by $\vert w\vert$ its length (with $\vert\emptyword\vert = 0$) and by $\Sigma^n$ the set of strings of length $n$ over $\Sigma$; by $\Sigma^{\leq n}$ we denote the strings of length at most $n$ over $\Sigma$. A subset $L \subseteq \Sigma^*$ is called a language. % For a finite language $\mathcal{F}$, let $\vert\vert\mathcal{F}\vert\vert=\sum_{w\in \mathcal{F}} \vert w\vert$ the total length of the strings in $\mathcal{F}$. 
The notation $w{[i]}$ is used to denote the symbol at position $i$ of $w$. Let $n,m \in \mathbb{N}$ be a pair of positive integers such that $n \geq m$. The string $u \in \Sigma^m$ is a \emph{factor} of $w \in \Sigma^n$ if and only if there exists an index $i \in [n - m]$ such that $w{[i]} w{[i + 1]} \dots w{[i + m - 1]} = u$; in that case, we denote $u=w{[i,i+m-1]}$. A factor $w{[i,j]}$ of $w \in \Sigma^n $ is a prefix (resp., a suffix) of $w$ if $i=1$ (resp., $j=n$).\looseness=-1

A deterministic finite automaton (DFA) is a construct $A=(Q,\Sigma, q_0,F,\delta)$, where $Q$ is a finite set of states, $\Sigma$ is the input alphabet, $q_0$ is the initial state, $F$ is the set of final states, and $\delta: Q\times \Sigma \rightarrow Q$ is the transition function. 

A DFA $A$ can be seen canonically as a labelled, directed multi-graph, denoted $G(A)$, with the set of vertices $Q$ and having an edge $(q_1,q_2)$ (so, from $q_1$ to $q_2$) labelled with $a$ if and only if there exists a symbol $a\in \Sigma$ such that $\delta(q_1,a)=q_2$. If there exists multiple symbols $a$ such that $\delta(q_1,a)=q_2$, then the graph has multiple edges between $q_1$ and $q_2$, each labelled with the corresponding symbol; however, due to the determinism of $A$, there are no two edges with the same label between the same two vertices. Using this interpretation of DFAs as graphs, the language accepted by the DFA $A$, denoted $L(A)$, is the set of labels of walks between the vertex $q_0$ and the vertices corresponding to final states (the label of a walk is obtained by concatenating, in order, the labels of its edges). Note that viewing DFAs as directed, labelled multi-graphs is a standard and widely used tool in formal languages, and allows for a very intuitive way of seeing strings of $L(A)$ as walks in the graph $G(A)$ and, vice versa.

The class of languages accepted by DFAs is exactly the class of regular languages.

A state in $q \in Q$ is a \emph{failure state} if no walk (including the empty walk) starting in this state ends in a final state. In a standard way, one can detect the failure states of an automaton in linear time w.r.t. the size of the automaton. If these failure states (and the edges incident to such a state) are eliminated from $A$, we obtain an {\em incomplete} DFA (meaning that the transition function becomes partial). Moreover, an incomplete DFA $A=(Q,\Sigma, q_0,F,\delta)$ can be cannonically completed by adding a single failure state $q$ and setting all the undefined transitions $\delta(r,a)\gets q$, for $r\in Q$ and $a\in \Sigma$, as well as $\delta(q,a)\gets q$, for all $a\in \Sigma$.

For more details on DFAs and the languages accepted by them, see \cite{HopcroftUllman}. \looseness=-1

In this paper, we are interested in \emph{prefix-closed regular languages} (denoted, for short, PCL-languages). A regular language $L$ is prefix-closed if $s{[1,i]}\in L$, for any string $s\in L$ and $i\in [\vert s\vert]$; that is, all prefixes of the strings in $L$ are also in $L$. As an alternative characterization, a regular language $L$, recognised by the complete DFA $A=(Q,\Sigma, q_0,F,\delta)$, is a PCL if and only if every state in $Q$ is either a final state or a failure state. Alternatively, a regular language $L$, recognised by the incomplete DFA $A=(Q,\Sigma, q_0,F,\delta)$, is a PCL if and only if every state in $Q$ is final.\looseness=-1

Examples of PCLs are given in Section \ref{app:singleFF}, and this class of languages is also discussed in \cite{CevorovaJMPS14,KaoRS09} and the references therein. 

In this paper, we assume that a PCL $L$ is always given as an incomplete DFA $A=(Q,\Sigma, q_0,Q,\delta)$ with all states final, called \emph{prefix-closed finite automata} (PCA, for short). All the walks in the graph corresponding to the PCA $A$ go through final states only, and those walks starting in the initial state of a PCA correspond bijectively to strings of the language accepted by $A$. This property is fundamental for our algorithms: enumerating the strings of length $m$ accepted by a PCA $A$ is equivalent to enumerating the labelled walks of length $m$ in the multi-graph $G(A)$ corresponding to the automaton $A$. \looseness=-1

\subsection{Languages of Strings with Forbidden Factors}\label{app:singleFF}

An interesting example of PCL is the following: Let $\mathcal{F}$ be a finite set of strings over an alphabet $\Sigma$, called {\em forbidden strings}. The language $L_{\mathcal{F}}$ of the strings over $\Sigma$ that do not contain any forbidden string (i.e., a string in ${\mathcal F}$) as a factor is a PCL-language. We can efficiently construct a PCA accepting it, based on the standard Aho-Corasick automaton recognising all strings ending with elements of ${\mathcal F}$ as shown in Lemma \ref{lem:Aho_Corasick}.

\begin{lemma}\label{lem:Aho_Corasick}
Given finite set of strings $\mathcal{F} \subseteq \Sigma^*$ (called the set of forbidden strings), we can build a PCA accepting $L_{\mathcal{F}}$ in $O(\vert \Sigma \vert \sum_{ w \in \mathcal{F}} \vert w \vert )$ time.
\end{lemma}

\begin{proof}
Let $A$ be the DFA constructed from the Aho-Corasick machine \cite{aho1975efficient} accepting all strings over $\Sigma$ ending with strings from $\mathcal{F}$. The DFA $A$ can be constructed in $O(\vert\vert\mathcal{F}\vert\vert  \vert\Sigma\vert )$ time. We can convert $A$ into an incomplete DFA $A'$ that recognises strings that do not contain any string from $\mathcal{F}$ as a factor by first removing all final states of $A$ and then making all other states final. Every transition from $A$ that did not involve a final state is then copied to $A'$.\looseness=-1
\qed \end{proof}

 Further, let us consider the language $L_{\mathcal F}$ in the case of a single forbidden factor (i.e., $\mathcal{F}=\{f\})$. In that case, a succinct representation of the incomplete DFA accepting $L_{\mathcal{F}}$ can be constructed more efficiently, in $O(m)$ time only, where $m=\vert f\vert$. The key observation (see, e.g., \cite{CliffordJPS12}) is that the DFA $A=(Q,\Sigma,q_0,F,\delta)$ accepting all strings over $\Sigma$ ending with $f$ has the following structure:\looseness=-1
\begin{itemize}
    \item $Q=\{q_0,q_1,\ldots,q_{m}\}$ and $q_{m}$ is the single final state;
    \item $\delta(q_{i-1},f_i)=q_{i}$ for $i\in \{0,1,\ldots, m-1\}$ and for all $i\in \{0,1,\ldots, m\}$ and $a\neq f_{i}$ we have that $\delta(q_{i-1},a)=q_j$ for some $j\leq i-1$;
    \item the number of transitions connecting non-initial states of this automaton is in $O(m)$. 
\end{itemize}
Therefore, to construct $A$, it is enough to compute the transitions connecting its states to other non-initial states (as all other transitions lead to $q_0$). This can be done in $O(m)$ time using, e.g., the Knuth-Morris-Pratt algorithm \cite{KnuthMP77}. So, indeed, a succinct representation of the PCA accepting $L_{\mathcal{F}}$ can be constructed in $O(m)$ time (again, by removing the final state of $A$, making all other states final, and noting that all undefined transitions lead to $q_0$, except for the transition $\delta(q_{m-1},f_{m})$ which leads to a failure state).

Obviously, $L_{\mathcal F}$ is a PCL even if ${\mathcal F}$ is regular but not finite; $L_{\mathcal F}$ remains prefix closed (but not necessarily regular) for any language ${\mathcal F}$. For completeness, let us note here that there are PCLs which cannot be defined as $L_{\mathcal{F}}$ for some set ${\mathcal{F}}$. Indeed, the set of strings $L=\{w\mid w$ is a prefix of the infinite string $(abc)^\infty \}$ is a PCL, but it cannot be defined as $L_{\mathcal{F}}$ for some set ${\mathcal{F}}$, as any language $L_{\mathcal{F}}$ is also, e.g., suffix closed, and this is not the case of $L$. 

\subsection{Algorithm Sketch.}
% Before presenting the main technical details of our algorithm, we provide a high-level sketch of our approach to allow the reader a better understanding of the remainder of this paper.
The key idea behind our approach is to enumerate the set of walks via an \emph{implicit} yet \emph{meaningful} representation of the walks within $G$, i.e. a representation that does not require the walk to be explicitly output but allows the explicit representation to be retrieved in linear time. We do so by creating a pseudoforest $D(G)$ from $G$, which has the same vertices as $G$. Moreover, each vertex has at most one outgoing edge, called \emph{default edge}, corresponding to the first edge on (one of) the longest walk(s) leaving that vertex in $G$. We refer to $D(G)$ as the \emph{default graph} and note that there is a unique (potentially infinite) walk leaving the vertex $v$ in $D(G)$. The walks in $D(G)$ are called \emph{default walks}. Using the default graph, we can succinctly represent default walks as tuples: the unique default walk of length $\ell$ starting with $v$ is represented as $(v, \ell)$.\looseness=-1

When solving Problem \ref{prob:enum}, we maintain the most recently enumerated walk $\pi$ (which is a walk in $G$) as a list $(v_1, \ell_1) (u_1, v_2) (v_2, \ell_2) \dots, (u_{k - 1}, v_{k}) (v_k, \ell_k)$ where $(v_i, \ell_i)$ is the default walk starting in $v$ of length $\ell_i$, ending at $u_i$, and $(u_i, v_{i + 1})$ is a non-default edge of $G$ (so, not an edge of $D(G)$). To modify this walk $\pi$, and continue the enumeration, we find the last vertex $v'$ in the current walk $\pi$ that has a \emph{branch}, i.e., a non-default edge, $(v', u')$, which has not been considered yet as a continuation for the prefix of $\pi$ which connects $v_1$ to $v'$, and which starts at least one walk as long as the suffix of the current walk $\pi$ which connects $v'$ to $u_k$. Once $v'$ has been identified, a new walk is constructed, represented by $(v_1, \ell_1) (u_1, v_2) (v_2, \ell_2) \dots, (u_{i - 1}, v_{i}) (v_i, \ell_i ') (v' , u') (u', \ell_{i + 1}')$, i.e. a walk sharing the first $m - (\ell_{i + 1}' + 1)$ edges with $\pi$, followed by the non-default edge from $v'$ to $u'$, then the default walk from $u'$ of corresponding length. Importantly, this new walk can be enumerated by only changing at most three entries in the list representing $\pi$: first, updating the tuple $(v_i, \ell_i)$ to $(v_i, \ell_i')$; second, adding the edge $(v', u')$; third, adding the tuple $(u', \ell_{i + 1}')$. In this way, we require only a constant number of steps to update the walk.\looseness=-1

The main challenge of this algorithm is determining these branches in constant time while keeping the preprocessing linear. In Section \ref{sec:data_structures} we define the data structures used to achieve this result and show that these can be built in $O(\vert E \vert )$ time, where $\vert E \vert$ is the number of edges in the input graph. Section \ref{sec:enumeration} formalises the algorithms, and explains why the worst-case delay between outputs is $O(1)$.\looseness=-1

\section{Toolbox: Default Graphs and Data Structures}
\label{sec:data_structures}

\paragraph{Default edges, default walks, default graphs: definitions and basic facts.}
For the remainder of this paper, we consider the directed graph $G = (V, E)$, where $V = \{v_1, v_2, \dots, v_n\}$ is a set of $n$ vertices, and $E \subseteq V \times V$ is a set of directed edges represented by ordered pairs of vertices $(v, u)$.
%represents an edge in the graph from $v$ to $u$. 
We assume that we store a list of all outgoing and incoming edges for every vertex $v \in V$. We now introduce the primary data structures that are used for our enumeration algorithm.\looseness=-1

Firstly, observe that the longest walk starting with vertex $v$ either has length at most $n$ or infinite length. Hence, we compute, for each vertex $v \in V$, the length $\pi(v) \in [n] \cup \{\infty\}$ of the longest walk starting at $v$. Moreover, we compute and store an ordered list $L_{v}$ for each vertex $v \in V$, containing the pairs $((v, u), \ell)$, where $\ell\in [n]\cup\{\infty\}$ is the length of the longest walk from $v$ starting with the edge $(v, u)$. The list $L_{v}$ is ordered in decreasing order of the length-component of its elements, with ties broken according to the ordering of the target vertices of the edge-component of these elements, as induced by the ordering on $V$. We can show the following Lemma.

\begin{lemma}
    \label{lem:max_walk_length}
    Given a directed graph $G = (V, E)$, the lengths $\pi(v)$ and the lists $L_{v}$, for all $v \in V$, can be computed in $O(\vert E \vert )$ time.
\end{lemma}

\begin{proof}
The following algorithm is rather standard (and could be considered folklore), so we go quickly over it.

We first detect in $O(\vert V \vert + \vert E \vert )$ the strongly connected components of this graph, using, e.g., Tarjan's algorithm \cite{TarjanSC}. Then, we detect all the other vertices of $G$ from which there is a walk to a vertex from one of the strongly connected components. This, e.g., can be done as follows:
\begin{itemize}
    \item Initially, we colour all the vertices of $G$ white.
    \item We put all the vertices of the strongly connected components in a queue $R$ and also colour them red. 
    \item We repeat the following step until $R$ is empty: we extract the first vertex $p$ in the queue, and put in $R$ all the white vertices $s$ such that there is an edge $(s,p)$ in $G$, and colour each such vertex $s$ red. 
\end{itemize}
After this, we set $\pi(v)=\infty$ for all the red vertices, and remove them from $G$, as well as all the edges leaving or entering them. This whole process can be implemented in $O(\vert E \vert )$ time.

The remaining white vertices and edges form a directed acyclic graph $G'$. Now, we run the following folklore procedure: 
\begin{itemize}
    \item Find a topological ordering of the vertices of $G'$.
    \item Consider the vertices in inverse order of the topological ordering. When $v$ is considered in this order, set $\pi(v)$ to be $0$ if there is no edge leaving $v$ or, otherwise, $1$ plus the maximum of the values $\pi(u)$ where $(v,u)$ is an edge in $G'$.
\end{itemize}

It is immediate that, at this point, we have computed the length of the longest walk leaving each vertex $v$ and stored it in $\pi(v)$. Again, this whole process runs in $O(\vert E \vert )$ up to this point.

Further, we explain how the lists $L_v$ are computed, for all vertices $v$. Initially, all these lists are empty. Then, for each edge $(v, u)$, we store the tuple $((v, u), \pi(u))$ in a set $S$. As $\vert s\vert\in O(\vert E \vert )$, we sort (in $O(\vert E \vert )$ time, using radix sort) the triples of $S$ decreasingly according to the third component, breaking ties arbitrarily. Now, we go through the elements of the decreasingly sorted list $S$ and when the element $((v, u), \pi(u))$ is reached we add the pair $((v,u),\pi(u)+1)$ as the last element of the list $L_{v}$. This process computes correctly the lists $L_{v}$, requiring $O(\vert E \vert )$ time.

This concludes the proof of this lemma. \qed
\end{proof}

Now, for each vertex $v$, the first edge $(v,u)$ of the list $L_v$ is the first edge on (one of) the longest walk(s) starting with $v$. This first edge of the list $L_v$ is called in the following \emph{default edge} of $v$ (note that, by the definition of $L_v$, the notion of default edge is unambiguous, although there might be more than one longest walk starting with $v$). At this point, it is also important to note that, for each edge $(v,u)$, the longest walk from $v$ starting with $(v, u)$ has a length equal to the length of the longest walk starting with $u$ plus $1$. Consequently, there is a walk from $v$ starting with the edge $(v, u)$ whose length is maximal among all walks starting with $(v,u)$ and whose second edge is the default edge of $u$.\looseness=-1

Further, we define the \emph{default graph} $D(G) = (V, D(E))$. $D(G)$ has the same vertex set $V$ and the edge set $D(E)$ containing exactly the default edges of the vertices from $V$. By Lemma \ref{lem:max_walk_length}, as the default edges of $G$ (that is, the edges of $D(G)$) can be retrieved by simply taking the first edges of each of the lists $L_v$, with $v\in V$, we get that $D(G)$ can be constructed efficiently, in $O(|V|)$ time. Moreover, because each vertex $v$ of the default graph $D(G)$ has at most a single outgoing edge (the default edge of $v$), this graph is a pseudoforest, consisting of a collection of {\em default components}: disjoint cycles, trees whose roots are on cycles, and, respectively, independent trees, which are not connected to any cycle. In the trees of this collection, the orientation of the edges is induced by the orientation of the edges in $G$, i.e., from children to parents, so from the leaves towards the root. For an example, see Figure~\ref{fig:component_overview}.\looseness=-1

The walks of the graph $D(G)$ are called \emph{default walks}. We use the notation $(v, \ell)$ to represent the unique default walk of length $\ell$ starting from $v$.\looseness=-1

\begin{figure}
    \centering
    \vspace{-0.5cm}
    \begin{tabular}{r|l}
    \includegraphics[scale=0.3]{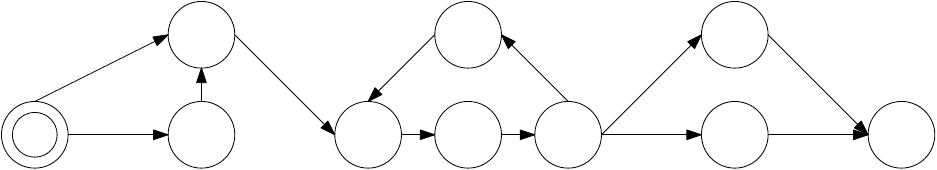} & 
    \includegraphics[scale=0.3]{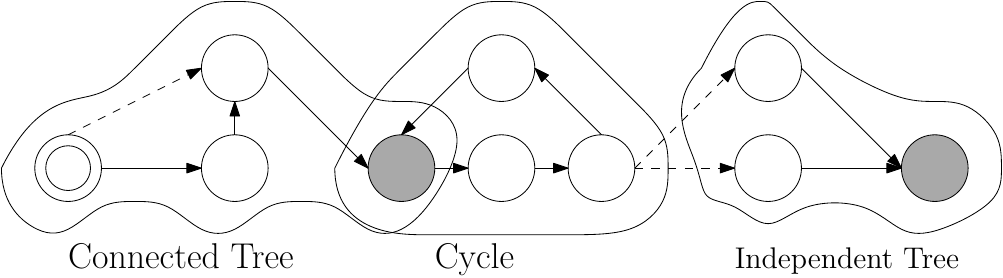}
    \end{tabular}
    \caption{The different classes of default components:
    On the left we have a directed graph; on the right, for that graph, default edges are shown as solid lines, and non-default edges are dashed. From left to right, we have a tree that is connected to a cycle, a cycle, and an independent tree. The tree-roots are grey. \looseness =-1}
    \vspace{-0.5cm}
    \label{fig:component_overview}
\end{figure}

We use the default walks as a tool to represent all walks in $G$. More precisely, given a walk $\pi = ((v_1, v_2), (v_2,v_3), \ldots , (v_{k-1},v_k))$, we represent the walk $\pi $ as a sequence $ (v_{i_1}, \ell_{i_1}) (v_{t_1}, v_{i_2}) (v_{i_2}, \ell_{i_2}) (v_{t_2}, v_{i_3}) \dots (v_{t_{r - 1}}, v_{i_r}) (v_{i_r}, \ell_r)$, where default walks and non-default edges alternate:
\begin{itemize}
    \item $v_{i_1}=v_1$ and $v_k$ is the final vertex of the default walk $(v_{i_r}, \ell_r)$. 
    \item For $1\leq x\leq r$, $(v_{i_x},\ell_{i_x})$, with $\ell_x\geq 0$, is the longest default walk which is a prefix of the suffix $(v_{i_x},v_{i_x+1},\ldots,v_k)$ of $\pi$. The vertex $v_{t_x}$ is the last vertex on the default walk $(v_{i_x},\ell_{i_x})$.
    \item For $1\leq x\leq r-1$, $(v_{t_x},v_{i_{x+1}})$ is a non-default edge and $i_{x+1}=t_x+1$.
\end{itemize}    
Alternatively, to obtain this representation of $\pi $, we could first select all the non-default edges of $\pi$. These edges are connected, along $\pi $, by default walks (of length greater or equal to $0$). This yields the aforementioned representation.

Before defining the data structures allowing us to process efficiently default graphs, we make one more observation. Consider a vertex $v$ and the first edge $(v,u)$ of $L_v$ (i.e., the default edge of $v$). Now, the first edge $(v, u')$, with $u \neq u'$, on the longest walk starting from $v$ with any edge other than $(v,u)$, is given by $((v, u'), \ell')$, the second element of $L_{u}$. So, these lists enable us to decide in constant time if, for a given length $\ell \leq m$ and vertex $v$, there exists some walk from $v$ of length at least $\ell$, other than the default walk: we get this information by looking at the second element of $L_v$. This will become crucial in finding branches, as described in the sketch of our algorithm, during the enumeration. 
\looseness=-1

\paragraph{Efficient algorithms and data structures for the default graph.} To work efficiently with the representations of arbitrary walks based on default walks, we need to be able to efficiently process the default graph. To this end, we now present a set of combinatorial lemmas, tools, and data structures providing a deeper understanding of the default graph. We begin by showing that the components of $D(G)$ can be computed efficiently.
% A full proof can be found in Appendix \ref{app:data_structures}.
\looseness=-1

% Thus, we need a series of lemmas, tools, and data structures, providing a more profound understanding and better structuring of this graph. We begin with the following result, which shows that the structure of $D(A)$ can be efficiently computed. A full proof can be found in Appendix \ref{app:data_structures}.

\begin{lemma}\label{lem:components}
Given default graph $D(G)$, we can compute in $O(\vert E \vert )$ time all its default components, and store for each vertex $v$ the default component containing it.\looseness=-1
\end{lemma}

\begin{proof}
We first identify in $D(G)$ the vertices $v$ with out-degree $0$ (that is, the vertices with $\pi(v)=0$). These are the roots of the independent trees. For each such vertex $v$ we perform a depth-first search in $D(G)^{-1}$ (the graph $D(G)$ with the orientation of the edges inverted), allowing us to discover the independent tree rooted in $v$. This process takes time proportional to the number of edges in these trees, so $O(n)$ overall, and we can associate to each vertex in the trees a pointer to the default component (e.g., a pointer to the root of the independent tree) which contains it.

To discover the cycles, we consider a vertex $v$ which does not belong to any independent tree. This is either on a cycle or in a tree whose root is on a cycle. In all cases, we can start a traversal of the graph following the only edge leaving $v$ (and, further, the single edges leaving the vertices we meet in this traversal). Following this walk will reach some vertex a second time, allowing us to identify a cycle of $D(G)$, denoted $\alpha$. If $v$ is the first vertex reached a second time, then $\alpha$ contains $v$, otherwise $\alpha$ is the cycle that is the root of the tree containing $v$. This cycle $\alpha$ is stored by arbitrarily choosing an initial vertex for $\alpha$, and, for each of the remaining vertices, we store pointers to $\alpha$ (e.g., to its initial vertex) as well as their position on $\alpha$ (w.r.t. the initial vertex). Then, for each vertex $v_r$ of $\alpha$ we discover the tree rooted at $v_r$, using the same techniques as for independent trees, with the single difference being that we do not explore in our traversal of the vertex of the cycle $\alpha$ from which there is an incoming edge towards $v_r$. In this way, we find the trees rooted in each of the vertex $v_r$ on the cycle, and compute all the required information for their vertices. The time needed to do this is $O(t_\alpha)$ where $t$ is the total number of vertices of $\alpha$ and of the trees rooted in vertices of $\alpha$ (as the time needed to complete the traversals we performed in our algorithm is proportional to the number of edges in the cycle and the attached trees). After we are done, if there are still vertices that do not belong to the default components discovered so far, we select one of them and repeat the process.  

Note that, for simplicity, we will assume that if a vertex is both on a cycle and the root of a tree, then it has pointers both to the cycle and to the tree structure.

This algorithm finds the component containing every vertex in $V$, and requires time proportional to the number of edges in the graph, hence $O(\vert E \vert )$. \qed\end{proof}

Trees appearing as default components of $D(G)$ (independent or attached to a cycle) can be represented as the root, followed by a list of children for each vertex; additionally, the default edge of each vertex points to its parent in the tree. Each cycle $\alpha$ is represented by its length $\vert\alpha\vert$ and an array containing $2\vert\alpha\vert$ elements. More precisely, for each cycle $\alpha$, we have an initial vertex $r$, which is the first element in the array corresponding to $\alpha$. We then traverse through $\alpha$ twice, while writing in the array the vertices of $\alpha$ in the order we meet them in this traversal. So, each element of $\alpha$ appears exactly twice in the array associated with the cycle, with exactly $\vert\alpha\vert -1$ positions between its two occurrences. For each vertex $v$ of $\alpha$, it is enough to store its first occurrence $i_v$ in the corresponding array. While a bit cumbersome at first view, this representation of cycles makes a bit simpler the usage of default graphs in the enumeration. \looseness=-1

% We have noted that our representation of walks in $G$ relies on the usage of pairs $(v, \ell)$ defining the default walk of length $\ell$ starting at vertex $v$. In the following, we provide a deeper analysis of these walks.

If vertex $v$ is on a cycle $\alpha$ of length $\vert\alpha\vert$, the default walk $(v,\ell)$ corresponds to following $\alpha$ starting at $v$ for $\ell$ edges. Letting $v$ be the $i^{th}$ vertex on $\alpha$ (the initial vertex being the first), then the ending vertex of the default walk $(v,\ell)$ is the vertex found on the $((i+\ell) \mod \vert\alpha\vert )^{th}$ position of the cycle.\looseness=-1
% easy to retrieve: if $v$ is the $i^{th}$ vertex in $\alpha$ (the initial vertex being the first), then the ending vertex of the default walk $(v,\ell)$ is the vertex found on the $((i+\ell) \mod \vert\alpha\vert )^{th}$ position of the cycle.\looseness=-1

If vertex $v$ is in an independent tree rooted at $r$, the default walk $(v,\ell)$ goes towards the root of 
the tree, traversing $\ell$ edges. % In other words, the walk goes through $\ell$ levels of the tree towards the root; 
% Recall that the levels of the tree are defined as follows:
We define the root as being on level $0$, and the children of a vertex on level $i$ are on level $i+1$. Therefore, if $v$ is on level $h$ in its tree, the ending vertex of the default walk $(v,\ell)$ is the ancestor of $v$ on level $h-\ell$ of this tree. Thus, to be able to retrieve the ending vertices of default walks in trees quickly, we build, for all our trees (both independent and attached to cycles), {\em level ancestor} data structures \cite{BenderF04}. For a tree of size $\tau$, these data structures can be computed in $O(\tau)$ time and enable us to answer in $O(1)$ queries $LA(v,j):$ return the ancestor of $v$ which is on level $j$ of the tree.\looseness=-1

If $v$ is in a tree whose root $r$ is on a cycle, the walk $(v,\ell)$ goes towards the root of the tree and potentially also goes around the cycle, traversing $\ell$ edges in total. Following the ideas already described above, to retrieve the ending vertex of a default walk $(v,\ell)$ in a tree with root $r$ attached to a cycle, we will first check if the walk ends inside the tree or enters the cycle. This can be done by verifying if the level $h$ of $v$ inside its tree is greater or equal to $\ell$. If yes, we can compute again the level ancestor of $v$, which is on level $h - \ell$. If not, then the ending vertex of the walk $(v,\ell)$ is on the cycle. The number of edges traversed in the cycle by this walk is $\ell-h$, and to find the ending vertex of $(v,\ell)$ it is enough to find the ending vertex of the default walk $(r,\ell-h)$, which is a walk on a cycle, and can be treated as above.\looseness=-1

According to the above, it is important to store, for the vertices of the tree components of $D(G)$, their level. To allow a uniform treatment of all the vertices, we define the \emph{depth of vertex $v$ in its default component}, denoted $d^t_{v}$, if $v$ is in a tree, or, respectively, $d^c_{v}$, if $v$ is on a cycle. This is defined in one of three ways, depending on which kind of component contains $v$. Before giving the definition, we note that the vertices contained in trees whose root is on a cycle will have two such depths, one w.r.t. the tree and one w.r.t. the cycle.\looseness=-1
\begin{itemize}
    \item If $v$ is in a tree, then $d^t_{v}$ is simply the level of $v$ in the respective tree.
    \item For a cycle $\alpha$, we define the depth of each position $i\leq 2\vert\alpha\vert$ of the array associated to $\alpha$ as $d^c_i=2 \vert \alpha \vert - i + 1$. If $v$ is a vertex on $\alpha$, we define $d^c_v=d^c_{i_{v}}$, where $i_{v}$ is the first (i.e, leftmost) occurrence of $v$ in the array associated to $\alpha$.\looseness=-1
    \item If $v$ is in a tree with root $r$ connected to a cycle $\alpha$, then we associate to $v$ a second value $d^c_{v}$ (the depth of $v$ w.r.t. the cycle $\alpha$), defined as $d^c_{v}=d^t_{v}+d^c_{r}$.
\end{itemize}
%The algorithm presented in Section \ref{sec:enumeration} enumerates all walks of length $m$ in $G$, for some given integer $m$. Initially, we present the algorithm for enumerating all walks starting at a given vertex $v$, noting that the precomputation does not depend on the start vertex. This can be generalised by computing a list of all vertices with a default walk of length at least $m$, then repeating the enumeration for each vertex of that list.\looseness=-1

\paragraph{Using the default graph for enumeration.}
The algorithm presented in Section \ref{sec:enumeration} enumerates all walks of length $m$ in $G$ starting at a given vertex $v_0$. We note that the preprocessing does not depend on the choice of the vertex $v_0$, it can be done once for all vertices. Our approach can be then extended by selecting (at the time when the values $\pi(v)$ are computed) a list of all vertices with a default walk of length at least $m$, and then repeating the enumeration for each vertex of that list.\looseness=-1

An important primitive of this algorithm, allowing us to move in our enumeration from one walk $\pi_1$ to the next one $\pi_2$ (both starting with $v_0$), is to check whether there exists a vertex $v$ on some default walk $(s,\ell)$ (which is part of the representation of the first walk $\pi_1$) from which we can follow a non-default edge, instead of the default edge which we have followed in the walk $\pi_1$, and obtain a new walk starting in $v_0$ of length $m$. Such a vertex $v$ is called, for simplicity of exposure, a {\em branching vertex} w.r.t. the walk $\pi_1$ (notice, though, that some vertices might be branching w.r.t. some walks $\pi_1$ and not branching w.r.t. others, depending on the position in which they appear on these walks; however, we will only use this name when there is no danger of confusion). Thus, we need to check the existence of a vertex $v$ on $(s,\ell)$, such that, if the walk from $v_0$ to such a vertex $v$ along $\pi_1$ has length $\ell_{v}$, then there is a walk starting with a non-default edge of $v$, with length at least $m - \ell_{v}$. \looseness=-1

We achieve this by finding the vertex $v$ of the default walk $(s,\ell)$ maximising the sum between the length of the walk from $v_0$ to $v$ along $\pi_1$ and the length of the longest walk starting in $v$ with a non-default edge, and seeing if this sum is greater or equal to $m$. If this sum is not greater or equal to $m$, then there is no vertex with the desired properties on $(s,\ell)$. If the sum is greater or equal to $m$, the vertex $v$ has the desired properties, and thus we can use it next in our enumeration. To identify this vertex $v$, we note that $v$ is exactly that vertex for which the sum of the length of the walk from $s$ to $v$ and the length of the longest walk starting in $v$ with a non-default edge is maximum (as all walks from $v_0$ to vertices on the default walk $(s,\ell)$, along $\pi_1$, share the prefix of $\pi_1$ connecting $v_0$ to $s$, the starting vertex of that default walk $(s,\ell)$). \looseness=-1

We define for each vertex $v$ a weight $w_v$, corresponding to the length of the longest walk from $v$ which starts with a non-default edge, i.e. the length of the walk starting with the second edge stored in $L_v$, and then continuing with the longest default walk starting with the end vertex of that edge. After the preprocessing of Lemma \ref{lem:max_walk_length}, $w_v$ can be retrieved in $O(1)$ time, for each $v$.\looseness=-1

Let $\PMN(s,\ell)$ denote, for each default walk $(s,\ell)$ of $D(G)$, the pair $(v,d)$ where $v$ is the vertex of this walk such that the sum of $w_v$ and the distance $d$ between $s$ and $v$ along the default walk $(s,\ell)$ is maximum. In Lemma \ref{lem:max_node_on_walk}, we show that $\PMN$ queries can be answered in $O(1)$-time, after linear time preprocessing. Indeed, building on the data structures introduced above and taking into account the particular structure of the default graph, these queries reduce to either computing walk minimum queries in trees, which can be handled efficiently \cite{demaine2009cartesian}, or to range minimum queries in arrays corresponding to the cycles, which again can be computed efficiently \cite{BenderF00}.\looseness=-1

\begin{lemma}\label{lem:max_node_on_walk}
We can build in $O(\vert E \vert )$ time data structures, allowing us to answer $\PMN(s,\ell)$ queries in $O(1)$ time, for each default walk $(s,\ell)$ of $D(G)$.
\end{lemma}

\begin{proof}
Firstly, we will try to obtain a clearer image of which vertex we want to retrieve from each default walk.

Let us assume that we are given the default walk $(s,\ell)$. We have several cases, according to the type of the default component containing the respective walk. In each case, once we clarify what vertex we need to compute, we also explain what data structures are needed, and how they can be used to return this vertex and its distance from $s$.

The main idea is to define for each vertex of the graph $D(G)$ a value $f(v)$ which is independent of the walks we are given as queries, such that finding for some default walk $(s,\ell)$ of $D(G)$ the vertex $v=\PMN(s,\ell)$ of this walk, which maximises the sum of $w_v$ and the distance between $s$ and $v$ along the default walk $(v,\ell)$, becomes equivalent to finding the vertex $v$ of that walk, for which the value $f(v)$ is maximum.

{\bf Case 1.} Assume that $s$ is in an independent tree $\tau$. Then, the respective default walk $(s,\ell)$ is a walk going in the respective tree from $s$ towards its root $r$, and as explained before, we can retrieve its endpoint $s'$ in $O(1)$ time using one level ancestor query in $\tau$. We want to retrieve the vertex $v$ of this walk such that the sum of $w_v$ and the distance between $s$ and $v$ along the default walk $(s,\ell)$ is maximum. But this also means that the sum of $w_v$ and the distance between $s$ and $r$ in the tree, from which we subtract $d^t_v$, is maximum. In other words, we look for the vertex $v \neq s'$ on the unique walk from $s$ to $s'$, for which the value $f(v)= w_v-d^t_v$ is maximum; note that this value is independent of $(s,\ell)$. 

So, the setting is that we have the tree $\tau$ and associate to each of its vertices $v$ the value $w_v-d^t_v$. Now, we simply need to be able to retrieve the vertex of maximum value on walks starting from a vertex, having a given length, and going towards the root $r$. Hence, we need to answer {\em walk maximum queries in the tree} $\tau$. A data structure can be constructed in $O(n_\tau+T_\tau)$ time, allowing us to answer such queries in constant time \cite{demaine2009cartesian}. Here $T_\tau$ is the time needed to sort the values associated with the vertices of the tree. However, as these values can be sorted for all independent trees simultaneously, and the absolute values associated with the vertices of the trees are either $\infty$ or $O(n)$, then the total time needed to sort them is $O(n)$ (again, radix sort can be used, while keeping track of the tree from which each value comes). So, we can retrieve the vertex $v$ on the first component of a query $\PMN(s,\ell)$ in $O(1)$ time, after an $O(n)$-time processing; the distance $d$ between the start of the walk and the vertex $v$ is simply the difference between $d^t_s$ and $d^t_v$.

{\bf Case 2.} Assume now that $s$ is on a cycle $\alpha$. We can compute the end-vertex $s'$ of the walk $(s,\ell)$, as explained before. Now, let $i_s$ be the leftmost (first) occurrence of $s$ in the array associated with $\alpha$ and $i_{s'}$ (resp., $j_{s'}$) be the first (resp., second) occurrence of $s'$ in that array. If $\ell<\vert\alpha\vert$, then the walk from $s$ to $s'$ goes exactly once through the vertices contained between $i_s$ and $j_{s'}$ in the array associated to $\alpha$. Similarly to the case of trees, discussed above, it is again enough to identify that vertex $v$ for which $w_v-d^c_v$ is maximum (this value does not depend on $(s,\ell)$, just like above). If $\ell \geq \vert\alpha\vert$, then the walk from $s$ to $s'$ goes through all the vertices of the cycle one or more times. Now, if we consider a vertex $r$ on this cycle, it is not hard to see that the sum of $w_r$ and the distance between $s$ and $r$ along the default walk $(s,\ell)$ is maximum for the last occurrence of this vertex $r$ on the respective walk. Therefore, it is enough to select the vertex $q$ for which the sum of $w_v$ and the distance between $s$ and $v$ along the default walk $(s,\ell)$ is maximum by considering only the last occurrence of the vertices of $\alpha$ on the walk $(s,\ell)$. But, once more, just like in the case of trees, this means retrieving the vertex $v$ for which the value $w_v-d^c_v$ is maximum from the vertices appearing in the range $[i_{s'},j_{s'}-1]$. 

So, in both cases, we need to associate with each vertex $v$ in the array corresponding to $\alpha$ the value $f(v)=w_v-d^c_v$. Then, we build in $O(\vert\alpha\vert )$ time data structures allowing us to answer {\em range maximum queries for this array} in $O(1)$ time \cite{BenderF00}. That is, we can retrieve in $O(1)$ time the vertex $v$ of maximum value from a range $[g:h]$ of that respective array. The distance between $s$ and $v$ can be trivially computed in $O(1)$ time. This is precisely what we needed.

{\bf Case 3.} Assume, finally, that $s$ is in a tree $\tau$ whose root $r$ is on a cycle $\alpha$. We produce the same data structures as in case 1 for the tree $\tau$ and we already have the data structures mentioned in case 2 for the cycle $\alpha$. Now, we explain how to answer the queries for a walk $(s,\ell)$.

If the default walk $(s,\ell)$ is completely contained in $\tau$, we can proceed as in case 1. This can be checked by looking at whether the end vertex of the walk $(s,\ell)$ is a vertex of $\tau$. If $\ell > d^t_s$, then the walk ends inside the cycle $\alpha$. So, we split the default walk into two sub-walks $(s,d^t_s)$ and $(r, \ell-d^t_s)$ and obtain the vertex $v_1$ on the walk $(s,d^t_s)$ such that the sum of $w_{v_1}$ and the distance between $s$ and $v_1$ along the default walk $(s,d^t_s)$ is maximum (as in case 1) and the vertex $v_2$ on the walk $(r,\ell-d^t_s)$ such that the sum of $w_{v_2}$ and the distance between $r$ and $v_2$ along the default walk $(r,\ell-d^t_s)$ is maximum (as in case 2). Note that $v_2$ also has the property the sum of $w_{v_2}$ and the distance between $s$ and $v_2$ along the default walk $(s,\ell)$ is maximum compared to the respective sum for all the vertices on the walk $(r,\ell-d^t_s)$. Now, as the distance between $s$ and $v_1$ can be computed easily (it is $d^t_s - d^t_{q_1}$) as well as the distance between $s$ and $v_2$ (it is simply $\ell$ from which we subtract the distance between $v_2$ and the end of the walk $(r,\ell- d^t_{s})$), we can compare which is greater: the sum of $w_{v_1}$ and the distance between $s$ and $v_1$ along the default walk $(s,\ell)$ or the sum of $w_{v_2}$ and the distance between $s$ and $v_2$ along the default walk $(s,\ell)$. The vertex for which the respective sum is greater is the vertex we wanted to retrieve. It is important to notice that obtaining the respective vertex can be done in $O(1)$ (as well as the distance from $s$ to the respective vertex) after the linear preprocessing described at the beginning of the analysis for this case is performed. 

This ends our case analysis and concludes the proof of our claim. 
\qed\end{proof}
 
\section{Enumeration}
\label{sec:enumeration}

% {\colour{red} to update}

% Using the data structures introduced in Section \ref{sec:data_structures}, in particular default edges, walks, and components, we present our approach for enumerating the walks of length $m$ in $G = (V, E)$. We start with enumerating the set of walks starting at a given vertex $v_0$, before generalising this to the set of all walks.\looseness=-1
We can now formally describe our algorithm. Its input is a positive integer $m$, the directed graph $G = (V, E)$ with $n$ vertices, and the vertex $v_0$. We want to enumerate every walk of length $m$ starting at $v_0$.
% PCA $A = (Q,\Sigma,q_0,Q,\delta)$, with $\vert Q\vert=m$ and $\vert\Sigma\vert=\sigma$. We want to enumerate the elements of the language $L(A)\cap \Sigma^n$. 

We first preprocess the graph $G$ using the algorithms from Section \ref{sec:data_structures}, including those from Lemmas \ref{lem:max_walk_length}, \ref{lem:components}, \ref{lem:max_node_on_walk}. This takes $O(\vert E \vert )$ and we obtain all data structures defined in Section \ref{sec:data_structures}, including, in particular, data structures allowing us to retrieve in $O(1)$ time the answer to $\PMN$ queries. Note that this preprocessing time does not depend on $m$.

We also define two empty stacks ${\mathcal S}$ and ${\mathcal C}$. 

These are maintained as global variables during the execution of our algorithm. The stack ${\mathcal S}$ contains tuples $((v, u), \ell)$, with $(v, u) \in E, \ell\in [m]$. Intuitively, if the content of the stack ${\mathcal S}$ is, at some step of the computation the sequence $\langle(\uparrow,v_0,\ell_0), ((u_0, v_1),\ell_1), \ldots, ((u_{t-1},v_t),\ell_t)\rangle$ (where the top of the stack is to the right of this sequence), then the currently enumerated walk $\pi$ is $(v_0, \ell_0) (u_1, v_1) (v_1, \ell_1) \cdots (v_{t - 1}, \ell_{t - 1}) (u_{t - 1}, v_{t}) (v_t, \ell_t)$, where, for $i\geq 0$, $(v_{i}, \ell_i)$ is the default walk of length $\ell_i$ starting at $v_i$ and ending with the vertex $u_i$. The walk $\pi$ can be retrieved explicitly from its representation on the stack ${\mathcal S}$ in linear time. With respect to the execution of our algorithm, ${\mathcal S}$ corresponds to the stack of currently active recursive calls. 

The usage of ${\mathcal C}$ is more subtle. The stack ${\mathcal C}$ contains tuples $((v, u),\ell)$, with $(v, u) \in E, \ell\in [m]$, with the property that each such tuple also occurs in ${\mathcal S}$; therefore, for each tuple, we will also store, together with it, in the stack ${\mathcal C}$ a pointer to the corresponding record of ${\mathcal S}$. Intuitively, at every step of the computation, ${\mathcal C}$ contains (bottom to top, in the same order as in ${\mathcal S}$) exactly those tuples $((v, u),\ell)$ of ${\mathcal S}$ for which the default walk of length $\ell$ starting in $u$ still contains branching vertices leading to walks of length $m$ which were not enumerated yet. As such, ${\mathcal C}$ allows for the quick identification of the next walk which we need to output in our enumeration: such a walk should go through one of the branching points of the default walk found on top of this stack ${\mathcal C}$. From the point of view of the execution of our algorithm, ${\mathcal C}$ corresponds to the currently active recursive calls, which were not tail calls. Recall that we assume that, in the computational model we use, tail calls are implemented so that no new stack frame is added to the call stack. Hence, ${\mathcal C}$ actually corresponds, at each moment of our algorithm's execution, to the current call stack.

As mentioned in the informal description of our approach in Section \ref{sec:enumeration}, a recursive procedure, named $\ennumerate$, is central to our approach. At each call, $\ennumerate$ takes two parameters, an edge $ (v, u) \in E \cup\{\uparrow\}$, and a length $\ell$. Referring to our intuitive explanation, we now have to go through all walks starting from $u$ and having length $\ell$. Each will lead to a walk of length $m$ that we need to output.

Now we will describe how the call $\ennumerate((v, u),\ell)$ works. The pseudo-code of this procedure is given in Algorithm \ref{alg:enumerate}.

\begin{algorithm}[H]
    \caption{The main recursive procedure of our enumeration algorithm.}
    \label{alg:enumerate}
\begin{algorithmic}[1]
\Procedure{$\ennumerate$}{Edge $(v,u)$ of $G$, length $\ell$}
\State Push $((v, u),\ell)$ in ${\mathcal S}$.
\State Push (stack frame of this call, pointer to the top element of ${\mathcal S}$) \newline
 \hspace*{40pt} as the top element of ${\mathcal C}$.
\State Output $(m-\ell-1, (v, u), \ell)$, which describes the current walk \newline
\hspace*{40pt} w.r.t. the previous one.
\State Return, if $\ell=0$.
\State Let $U$ be an empty list.
\State Let $(v',d)=\PMN(u,\ell)$. 
\If{$(d+w_v)\geq \ell$}
	\State Add $((u,\ell),u,0)$ to $U$.
	\State Let $((v', u'),e)$ be the second element of $L_{v'}$.
	\State Let $last=((v', u'),\ell-d-1)$. 
\Else
	\State Remove the elements of  ${\mathcal S}$ and ${\mathcal C}$ corresponding to this call.
	\State Return.
\EndIf
\While{$U$ is not empty}
	\State Extract the first tuple $((s,j),u,h)$ from $U$. Let $(v',f)=\PMN(s,j)$. 
	\State Add $((s,f),u,h)$ to $U$, \newline
            \hspace*{40pt} if $f>0$, $(r,e)=\PMN(s,f)$ and $h+e+1+w_r\geq \ell$. 
	\State Add $((v'',j-f-1),u,h+f+1)$ to $U$, where $v''$= successor of $v'$ on $(s,j)$, 
        \newline 
        \hspace*{40pt}if $j-f-1>0$, $(r,e)=\PMN(v'',f)$ and $h+f+1+e+1+w_r\geq \ell$.
	\If{$((s,j),u,h)\neq ((u,\ell),u,0)$}
		\State $x=2$
	\Else
		\State $x=3$
	\EndIf
	\State Let $(b,g)$ be the $x^{th}$ element of $ L_{v'}$ ($(b,g)$ is undefined if $\vert L_{v'}\vert<x$)
	\While{$(b,g)$ is defined and $h+f+g +1\geq \ell$}
		\State $\ennumerate(b,v',\ell-h-f-1)$.
		\State Set the top element of ${\mathcal S}$ to the element of ${\mathcal S}$ pointed \newline
  \hspace*{7cm} by the top element of ${\mathcal C}$. 
		\State Let $(b,g)$ be the next element of $ L_{v'}$ ($(b,g)$ is undefined if $\vert L_{v'}\vert<x$)
	\EndWhile
\EndWhile
\State Pop the stack frame of $\ennumerate((v,u),\ell)$, and the attached pointer, from ${\mathcal C}$.
\State $\ennumerate(last),$ where $last=((v', u'),\ell-d-1)$
\State Return.
\EndProcedure
\end{algorithmic}\label{algo:code}
\end{algorithm}

We begin by pushing $((v, u),\ell)$ in the stack ${\mathcal S}$. Furthermore, an element containing a pointer to the element just pushed in ${\mathcal S}$ as well as the stack frame of the call $\ennumerate((v, u),\ell)$ is pushed in ${\mathcal C}$.\looseness=-1  

Then, we output the tuple $(m-\ell-1,(v, u),\ell)$, encoding walk of length $m$ consisting of the first $m-\ell-1$ edges of the previously output walk, followed by $(v, u)$, and the default walk $(u,\ell)$. 

At this point, we need to see whether there are more walks of length $\ell$ starting with $u$ that we need to go through. Each of these walks will be obtained via a recursive call of $\ennumerate$, and it is important to make sure that the last walk we discover in this way is obtained via a tail call. To facilitate the discovery of walks, we will maintain a list $U$, whose elements are of the form $((s,j),u,d)$, where $(s,j)$ is a default walk on which there is at least one branching vertex that needs to be explored, and $(s,j)$ is part of the default walk starting from $u$ such that one needs to traverse $d$ edges to get from $u$ to $s$ on this walk. $U$ is initially empty. 

Firstly, we detect if there is at least one walk of length $\ell$ from $u$, that we did not enumerate yet. For this, we retrieve $(v',d)=\PMN(u,\ell)$. 

If $d + w_{v'} \geq \ell$, we add the tuple $((u,\ell),u,0)$ to $U$; this means that there is a branching vertex on the default walk $(u,\ell)$ from which we can obtain at least one new walk. One of these new walks will be certainly obtained by calling $\ennumerate((v', u'),\ell-d-1)$, where $((v', u'),e)$ is the second element of $L_{v'}$, i.e., the longest walk starting with a non-default edge from $v'$, starting with the edge $(v', u')$ and continues with the default walk from $u'$. Note that the fact that this walk is long enough is guaranteed by the truth of the inequality $d + w_{v'} \geq \ell$. Now, we have discovered at least one recursive call that we need to do in the current instance of $\ennumerate$. However, we will not execute this call immediately. Instead, we will postpone this call until the end of the procedure, and this will be the tail call (avoiding a series of tedious checks which would allow us to manage the tail call otherwise). We save in a variable $last$ the tuple $((v', u'),\ell-d-1)$, the parameters of the future tail call.

If $d + w_{v'} < \ell$, then there are no more walks to be explored (the default walk from $u$ was the only one), so we simply return (and remove the elements corresponding to the call $\ennumerate((v, u),\ell)$ from both stacks ${\mathcal S}$ and ${\mathcal C}$).

In the case we added something to $U$, we continue as follows. 

While $U \neq \emptyset$, we extract the first tuple $((s,j),t,h)$ from $U$ and process it as follows. 

Assume first that $((s,j),u,h) \neq ((u,\ell),u,0)$.

We retrieve $(v',f)=\PMN(s,j)$. This is a branching vertex, so we will need to do some recursive calls with it. But first, we will update $U$ as follows. We add to $U$ the element $((s,f),u,h)$, if $f>0$ and $(r,e)=\PMN(s,f)$ and $h+e+1+w_r\geq \ell$. We add to $U$ the element $((v'',j-f-1),q,h+f+1)$, where $v''$ is the successor of $v'$ on the default walk $(s,j)$, if $j-f-1>0$ and $(r,e)=\PMN(v'',f)$ and $h+f+1+e+1+w_r\geq \ell$. The idea is that we will now explore the possible walks originating in $v'$ but we still need to explore the possible walks originating in branching vertices from $(s,f)$, if that walk is non-empty (meaning that $s\neq v'$ or, in other words, $f>0$), or from $(v'',j-f-1)$, if that walk is non-empty (meaning that $j-f-1>0$). Once $U$ is updated, we traverse the list $L_{v'}$ starting with its second element. Let the current element be $((v', u'),g)\in L_{v'}$. We check if $h+f+g +1\geq \ell$ (i.e., we check if there is a long enough walk going on the default walk $(u,\ell)$ through $u$, then $s$, until it reaches $v'$ and then following the edge $(v', u')$, which can lead to a new walk that we need to enumerate). If  $h+f+g +1\geq \ell$, we call $\ennumerate((v', u'),\ell-h-f-1)$; when this call is completed, we update the pointer to the top element of ${\mathcal S}$ to correspond to the element of ${\mathcal S}$ pointed by the top element of ${\mathcal C}$ (i.e., eliminate tail calls from ${\mathcal S}$, as now they also became useless in our representations of walks) and continue our traversal of $L_{v'}$. If If $h+f+g +1< \ell$, we simply stop the traversal of $L_{v'}$. 

If $((u,j),u,h)= ((u,\ell),u,0)$, note that $\PMN(u,\ell)=(v',d)$. We do exactly the same thing as above, except for the fact that we start the traversal of $L_{v'}$ with its third vertex. This is because the second vertex of $L_{v'}$ is $u'$, and we have already saved the call $\ennumerate((v', u'),\ell-d-1)$ to be the tail call of our procedure.

As soon as $U$ is empty, we will simply run the tail call $\ennumerate((v', u'),\ell-d-1)$, and then return. Just before doing that, we remove the stack frame of $\ennumerate((v, u),\ell)$ from the stack ${\mathcal C}$, to simulate the management of tail calls in our model. 

\begin{lemma}
    \label{lem:every_walk}
    The call $\ennumerate((\uparrow,v_0),m)$ outputs a representation of every walk of length $m$ starting at $v_0$ in the graph $G$  exactly once.
\end{lemma}

\begin{proof}
Note first that if we call $\ennumerate((v, u),\ell)$ at some point during the execution of our algorithm, then there is a walk of length $m-\ell$ ending with the edge $(v,u)$ (or the walk representing the empty walk if $v = \uparrow$) that leads from $v_0$ to $u$. This can be shown by induction on $m - \ell$. If $m - \ell = 0$, the conclusion follows by the fact that the only time we call $\ennumerate$ with the second parameter being equal to $m$ is in the first call $\ennumerate((\uparrow,v_0),m)$. Also note that no other call is made during the execution of $\ennumerate((\uparrow,v_0),m)$ which has the edge parameter $(\uparrow, v_0)$. Now, assume that the claim holds for $m - \ell \leq k$, and we will show it for $m - \ell = k+1$. In this case, the call $\ennumerate((v, u),\ell)$ was initiated by a previous recursive call $\ennumerate((v', u'),\ell+t)$, for some $t>0$. This means that the vertex $v$ is on the default walk $(u',\ell+t)$, at distance $t-1$ from $u'$, and the edge from $v$ to $u$ is non-default. The conclusion follows immediately.

By this claim, we immediately get that every walk that is output during the execution of the call $\ennumerate((\uparrow,v_0),m)$ represents a walk of length $m$ of $G$. 

We now show that during the call $\ennumerate((v, u),\ell)$ we correctly identify all branching vertices on the walk $(u, \ell)$; that is, all vertices $v'$ from which we can follow a non-default edge and extend the walk leading from $v_0$ to $u$ and then to $v'$ to a walk of length $m$. Firstly, for such a vertex to exist, then $d+w_{v'} \geq \ell$ must hold for $(v',d)=\PMN(u,\ell)$. If $(v',d)$ does not fulfil this requirement, there is no other walk of length $\ell$ starting in a vertex of the default walk $(u,\ell)$, and our algorithm reports that correctly. So, our algorithm considers $v'$ as the first branching vertex. If the walk $(u,\ell)$ had length $1$, then we are done. If our walk is longer, after correctly identifying the branching vertex $v'$, such that $(v',d)=\PMN(u,\ell)$, our algorithm looks for branching vertices on the default walks $(u,d)$ (from $u$ to $v'$) and $(v'',\ell-d-1)$, which starts with the successor of $v$ on the default walk $(u,\ell)$. Basically, we have partitioned the walk $(u,\ell)$ in three parts: the branching vertex and two shorter walks, on which we keep looking for branching vertices. The search on each of these walks is done exactly like the search on the whole walk $(u,\ell)$, with $\PMN$ queries, which is correct. Ultimately, we either rule out entire walks because not even the vertex returned by $\PMN$ leads to a long enough walk, or we split them further. Eventually, in at most $\ell$ splitting steps, each vertex of the default walk $(u,\ell)$ is either correctly ruled out or discovered as a branching vertex.

It only remains to show that each walk of length $m$ from $\mathcal{P}$ is only enumerated once, so we do not output two triples representing the same walk from $\mathcal{P}(m)$. This follows from the fact that a walk $\pi \in \mathcal{P}(m)$ has a unique decomposition as a concatenation of default walks and connecting non-default edges. As above, let $\pi= (v_0, \ell_0) (v_1, u_1) (u_1, \ell_1) (v_2, u_2) (u_2, \ell_2) \ldots (v_{t}, u_{t}) (u_t, \ell_t)$, where $(u_{i}, \ell_i)$ is the default walk of length $\ell_i$ between the vertices $u_{i}$ and $v_{i + 1}$, for $i\in \{0\}\cup [t]$, and $(v_{i}, u_i)$ is a non-default edge between $v_{i}$ and $u_{i}$, for $i\in [t]$. Basically, to obtain the walk $\pi$, we first call $\ennumerate(\uparrow,v_0,m)$. Then, the next call of $\ennumerate$ must be $\ennumerate((v_1, u_1), m - \vert\ell_0\vert -1)$ 
(there is no other way to choose a branching vertex than choosing the one at the end of the longest common prefix of the first walk $\pi$ and the default walk $(v_0,m)$). By an inductive process, we see that actually the calls of $\ennumerate$ that lead to the output of the representation of $\pi$ are uniquely determined. Therefore, a representation of this string is only output once.

This concludes our proof.
\qed\end{proof}

\begin{theorem}
\label{thm:complexity_fixed_length}
Given integer $m$, directed graph $G = (V, E)$, and vertex $v_0$, we can enumerate, without repetitions, succinct representations of all walks of length $m$ starting from $v_0$ in $G$ with $O(1)$-delay, after an $O(\vert E \vert )$-time preprocessing.\looseness=-1
\end{theorem}

\begin{proof}
We build, for the graph $G$, the data structures from Section \ref{sec:data_structures}, as a preprocessing, and then run $\ennumerate((\uparrow,v_0),m)$, if $\pi(v_0)\geq m$.

To show that the enumeration performed by the algorithm is done with constant delay, it is enough to show that the operations performed between two consecutive output operations can be executed in constant time.

Each output operation is associated to a call of $\ennumerate$. So, let us consider two such output operations, caused by two distinct calls to $\ennumerate$. There are several cases we need to check. 

In the first case, the first output operation is done in a (parent-)instance of $\ennumerate$, which then directly calls the instance of $\ennumerate$ in which the second output is performed. This means that the operations performed between these calls are those from lines 5--15 of Algorithm \ref{alg:enumerate} followed by exactly one execution of lines 17--25 from the while-loop of line 16; all these are done in the parent instance. With the data structures defined in Section \ref{sec:data_structures}, these can all be executed in $O(1)$ time.

In the second case, both output operations are done in instances of $\ennumerate$ originating (maybe not directly, but in a sequence of recursive calls) from the same parent instance of $\ennumerate$. The idea is that once the first output operation is performed, the instance that executed this operation will not make any other recursive call (as this would cause an intermediate output operation) before it ends. Once it ends (and, note at this point that getting from the output in line 4 to any return instruction requires only $O(1)$ time, with our data structures, if no further recursive calls are made), exactly the next call of $\ennumerate$ must be the one that causes the second output-operation (or, again, we would have intermediate output operations). So, the algorithm directly returns to the parent instance (the origin of the two calls leading to the two consecutive output operations), and this means that every intermediate call that led from the parent instance to the one causing the first output was a tail call. So, this return to the parent instance is done in $O(1)$ time in our setting. Now, looking at the recursive call that led to the instance of $\ennumerate$ which produced the first output and at the one that leads to the second output, it must hold that they happen in two consecutive iterations of the while-loop from line 26 of the algorithm, or the first happens in the last iteration of that loop, and the second is either the tail call from line 33 or the first call made in a new iteration of the while-loop from line 16 (and in the first iteration of the while loop from line 26 in this new iteration). In both cases, the time required to complete the computations between two such calls is constant (with the help of our data structures). 

To conclude, while the first case is straightforward, the second case is a bit more involved and can be seen like this. The first output is done by an instance of $\ennumerate$, which then returns after $O(1)$ time. This then takes us in $O(1)$ time (due to the management of tail calls) to the parent instance that originated the sequence of calls leading to the first output. Then, we can reach the call of $\ennumerate$ that makes the second output in $O(1)$ time. This is basically the next call to $\ennumerate$ made by the parent instance, after the one that led to the first output; the time needed to execute the instructions between two such calls is $O(1)$. 
Note also that after the final output, we have an empty stack and thus can determine in $O(1)$ time that no walks starting at $v_0$ that have not been enumerated, and therefore can terminate the algorithm.

There are no other cases to be considered, so our claim follows.
\qed\end{proof}

It is worth noting that the space used by our algorithm (on top of the $O(|E|)$-space used by the data structures produced during the preprocessing) is upper bounded, at any point, by the length of the currently enumerated walk. We can immediately extend Theorem \ref{thm:complexity_fixed_length} to output every walk of length $m$ in $G$: in the preprocessing phase, we collect every vertex with a default walk of length at least $m$, and then use the algorithm of Theorem \ref{thm:complexity_fixed_length} for each such vertex.\looseness=-1

\begin{theorem}
    \label{thm:any_starting_vertex}
    Given integer $m$ and directed graph $G = (V, E)$, we can enumerate, without repetitions, succinct representations of all walks of length $m$ in $G$ with $O(1)$-delay, after an $O(\vert E \vert )$-time preprocessing.
\end{theorem}

\begin{proof}
    Note that a list $Q = (v_1, v_2, \dots, v_n)$ containing all vertices in $G$ ordered by length of default walk can be constructed in $O(n)$ time via count sort. Therefore, by making the series of calls $\ennumerate((\uparrow, v_1), m),$ $\ennumerate((\uparrow, v_2), m),$ $\ldots,$ $\ennumerate((\uparrow, v_{\ell}), m)$, where $\ell$ is the largest index such that the default walk from $v_{\ell}$ has length $m$, every walk in $G$ of length $m$ is output exactly once. As there is only a constant delay between the last walk output by $\ennumerate((\uparrow, v_i), m)$, and the first output by $\ennumerate((\uparrow, v_{i + 1}), m)$, the delay of this algorithm is $O(1)$ as well, proving the theorem.
\qed\end{proof}

We can also enumerate for integers $\ell \leq m$, after exactly the same preprocessing (which is independent of the length of the enumerated walks), representations of the walks of $G$ with length between $\ell$ and $m$.\looseness=-1

\begin{theorem}
\label{thm:complexity_var_length}
Given two integers $\ell \leq m$ and directed graph $G = (V, E)$, we can enumerate, without repetitions, the walks of length at least $\ell$ and at most $m$ in $G$, in increasing order of their length, with $O(1)$-delay, after an $O(\vert E \vert )$-time preprocessing.
\end{theorem}

\begin{proof}
We build, for the graph $G$, the data structures from Section \ref{sec:data_structures}, as a preprocessing. This takes $O(\vert E \vert )$ time. Again, note that this phase does not depend on $\ell$ or $m$. Additionally, we build the list $Q$ ordering the vertices in $V$ by length of default walk from longest to shortest. We then run, for $i$ from $\ell$ to $m$, the procedure $\ennumerate((\uparrow,v),i)$ for every $v \in Q$ such that the default walk from $v$ has length at least $\ell$. Note that we can check in $O(1)$ time if $v$ has no walk of length $i$ by looking up the first element of $L_v$.
%(of course, if $L(A)$ contains strings of length $i$, which can be checked in $O(1)$ time by looking at the length of the default walk leaving $q_0$, and seeing if this length is at least $i$).
The delay when moving from one length to another is constant, as we always end the enumeration of the walks of some length with a tail call and start the enumeration for the next length with the default walk of that length, starting in $v$. 
\qed\end{proof}

\paragraph{Applications for Automata.}
% Following our motivating problem, 
The result of Theorem \ref{thm:complexity_fixed_length} can be immediately applied to prefix closed regular languages (PCLs), given by the prefix closed automata (PCA, incomplete deterministic finite automata with final states only) accepting them. For this, we represent the input PCA $A$ as a directed, labelled multi-graph, $G(A)$, and enumerate all walks of length $m$ starting at the vertex $v_0$ corresponding to the initial state in $A$. Our algorithm still works without any change because between two vertices of $G(A)$ we have at most one edge with a certain label (although we might have multiple edges), and all our data structures can be extended canonically to this setting. As there is a bijective correspondence between the walks in $G(A)$ and the strings of $L(A)$, the following theorem follows.\looseness=-1

\begin{theorem}
    \label{thm:enum_PCA}
    Given integer $m$ and PCA $\mathcal{A}$, we can enumerate, without repetitions, succinct representations of all strings of length $m$ of $L(\mathcal{A})$ with $O(1)$-delay after an $O(\vert \mathcal{A} \vert )$-time preprocessing.
\end{theorem}

\begin{proof}
As noted in Section \ref{app:defs}, a PCA $A=(Q,\Sigma,q_0,F,\delta)$ can be represented as a directed, labelled multi-graph $G(A)$. Moreover, as PCAs are assumed to be deterministic, $G(A)$ has the property that for each letter $a\in \Sigma$ and pair of states $q_1,q_2\in Q$, there exists at most one edge labelled with $a$ going from the vertex $q_1$ to the vertex $q_2$ in $G(A)$. This property allows us to use our algorithm defined for the enumeration of walks starting from a given vertex in directed graphs to enumerate the walks starting from the vertex $q_0$ in $G(A)$. 

Indeed, it is enough to redefine the lists $L_v$ to be, for each vertex (or state of $A$) $v \in Q$, the list of the triples $((v, u), a, \ell)$, where $\ell\in [n]\cup\{\infty\}$ is the length of the longest walk from $v$ starting with the edge $(v, u)$ in $G(A)$ and $a$ is the letter labelling the edge $(v,u)$. The list $L_{v}$ is ordered in decreasing order of the length-component of its elements, with ties broken firstly according to the letter-component (w.r.t. the order of the letters in $\Sigma$), and, secondly, according to the ordering of the target vertices of the edge-component of these elements, as induced by the ordering on $V$. This makes, again, the definition of default edges unambiguous. Hence, we can simply use the algorithm for walk-enumeration in directed graphs, but note that now, whenever a non-default edge is used or output, the letter labelling it should be made explicit (i.e., output as well).

The result then follows directly from Theorem \ref{thm:complexity_fixed_length}.
\end{proof}

Among PCLs, the class of languages $L_{\mathcal{F}}$, of the strings over $\Sigma$ that do not contain any \emph{forbidden factor} from a finite set of strings ${\mathcal F}$, is of particular interest. It remains open whether the results of Theorem \ref{thm:complexity_fixed_length} and \ref{thm:complexity_var_length} can be improved for such languages. However, when $\mathcal{F}=\{f\}$ the following result holds (according to the construction shown in Section \ref{app:singleFF}). This outperforms the algorithm of~\cite{RuskeyS00}.\looseness=-1

\begin{theorem}
\label{thm:one_word}
Given integer $m$ and string $f \in \Sigma^*$, we can enumerate, without repetitions, succinct representations of the strings of length $m$ over $\Sigma$ which do not contain $f$ as a factor, with $O(1)$-delay, after an $O(\vert f \vert )$-time preprocessing.
\end{theorem}

\section{Ranking and Unranking}
\label{app:ranking}

We recall the setting of the problem. The \emph{rank} of a string $w$ in a language is the number of strings smaller than $w$ in the language under some ordering. The ranking problem requires computing the rank of a given string $w$.

The unranking operation is the reverse of the ranking operation, taking a number $i$ as the input and asking for the string of rank~$i$. The unranking problem requires computing the string of rank $i$.

In both cases, we consider the order induced by the enumeration algorithm of Section \ref{sec:enumeration}, and we want to show that these two problems can be solved in polynomial time.

As before, we assume that we have a PCA $A = (Q,\Sigma,q_0,Q,\delta)$, where $\vert Q\vert = n$ and $\Sigma = \{1,2, \dots, \sigma\}$. 

We will keep the presentation in this section rather informal, as the technicalities are straightforward.

Recall the description from the main part of the paper, which is based on the enumeration algorithm. The main idea is that both ranking and unranking require identifying a walk in the tree of recursive calls of our $\ennumerate$ function with root $\ennumerate(\uparrow,q_0,m)$. 

In the case of ranking, we identify the walk corresponding to $w$, and the branching vertices occurring on it, and then count the total number of walks of length $m$ corresponding to the leaves of subtrees of recursive calls occurring to the left of this walk (assuming that the recursive calls made by an instance are ordered in the tree left to right according to their call-order). This can be done by running $\ennumerate(\uparrow,q_0,m)$ and simply performing only the recursive calls that correspond to branching vertices on the walk labelled with $w$, and retrieving the number of induced walks for those that should have been called before them.

In the case of unranking, one standard approach is to use the ranking procedure to determine the letters of the searched string one by one, or, more efficiently, and closer to what we have done here so far, one can again run $\ennumerate(\uparrow,q_0,m)$ and we make only those recursive calls which lead to the $i^{th}$ walk of length $m$, in the order of our enumeration. 

Let us go now into more details. 

We will use the following lemma. 
\begin{lemma}[Folklore]
    Let $2 \leq \omega \leq 3$ be the exponent for matrix multiplication and ${\tt w}$ be the size of the memory word in our model.
    Given PCA $A$ and integer $m$, we can compute the number of strings of length $\ell$ starting in a state $q$, for all $q\in Q$ and $\ell\leq m$, in $O\left(\frac{m^2  n^{\omega}\log \sigma}{\tt w}\right)$ time.
\end{lemma}
Since $A$ is deterministic, the number of walks of length $\ell$ between two states $q$ and $q'$ can be retrieved from the matrix $N^\ell$, where $N$ is the $n \times n$ matrix which contains at the entry corresponding to $q$ and $q'$ the number of transitions between $q$ and $q'$, for all states $q$, $q'$. Then, using $N^\ell$, with $\ell\leq m$, for each state $q$, we can compute the number of walks of length $\ell$ starting in $q$ in $O\left(\frac{m  n\log \sigma}{\tt w}\right)$ time. An additional factor $\frac{m\log \sigma}{\tt w}$ occurs in this complexity due to the time needed to perform arithmetic operations with big numbers.

Now, the ranking procedure works as follows. We run the PCA $A$ on the input string $w$ and obtain its decomposition $w=w_0a_0w_1a_1\ldots w_{t-1}a_{t-1}w_t$, where $w_{i}$ is the label of a default walk between the states $q_{i}$ and $q'_{i}$, for $i\in \{0\}\cup [t]$, and $a_{i-1}$ labels a non-default edge between $q'_{i-1}$ and $q_{i}$, for $i\in [t]$. We store this decomposition, as well as the states $q_{i}$ and $q'_{i}$, for $i\in \{0\}\cup [t]$. If $w$ is the label of the default walk of length $m$ starting in $q_0$, then the rank of $w$ is $1$. Otherwise, we run $\ennumerate(\uparrow,q_0,m)$  (without making any outputs) and use two integer variables $count$, set to $1$ initially, and $total$, set to $0$ initially; $count$ keeps track of how many of the non-default transitions from the walk with label $w$ were met in our sequence of recursive calls (the non-default edges actually correspond one-to-one to these calls), and $total$ keeps track of how many walks we have identified and counted already, that come before the one labelled with $w$ in the enumeration. In this process, each time a call $\ennumerate (a,q,\ell)$ should be made, we check first if it corresponds to the $count^{th}$ non-default edge of the walk labelled with $w$ (i.e., the transition from $q'_{count-1}$ to $q_{count}$, labelled with $a_{count-1}$). If yes, we perform that call. If not, we increase $total$ by the number of walks of length $\ell$ originating in $q$. After the $t^{th}$ call of $\ennumerate$, we simply return $total+1$ as the rank of $w$.

As our procedure $\ennumerate(a,q,\ell)$ might end up going through every branching vertex on $(q,\ell)$ before making the next recursive call, the overall number of steps performed by this algorithm is (once the preprocessing is done) is $O(m^2\sigma)$. However, one can also reduce this factor $m^2\sigma$ to $m(m+\sigma)$ by noting that there is exactly one branching vertex for which we need to explore all the transitions leaving it. Indeed, we can check first for each branching vertex the total number of walks of desired length leaving from it (instead of going through each transition individually) and only go through the transitions of that state individually if this check indicates that our recursive call should be done for one of the non-default edges leaving the respective vertex. 

The soundness of this approach follows from the fact that we basically use the non-default edges on the walk labelled by $w$ to traverse a root-to-leaf walk of the tree of recursive calls started by $\ennumerate(\uparrow,q_0,m)$, keeping track of the number of walks of length $m$ of $A$ which are discovered by calls occurring in the subtrees of the tree of recursive calls, which should have been done in the enumeration process before the calls we actually execute. 

\begin{theorem}
    \label{thm:ranking}
    {\bf Ranking}. Given PCA $A$ and string $w \in L(A)$ of length $m$, we can compute the number of strings accepted by $A$ which are output before $w$ in our enumeration algorithm in $O\left(\frac{n\log \sigma}{\tt w}(n\sigma + m (n^{\omega} + m + \sigma))\right)$-time, where $2 \leq \omega \leq 3$ is the exponent for matrix multiplication. 
    \label{thm:unrank} 
    {\bf U}. Given integers $i$ and $m$, and PCA $A$, we can compute the $i^{th}$ string $w$ of length $m$ output in our enumeration algorithm in $O\left(\frac{n\log \sigma}{\tt w}(n\sigma + m (n^{\omega} + m + \sigma))\right)$-time. 
\end{theorem}

For unranking, one can use essentially the same approach. This time, we are given as input a number $i$.  If $i=1$, we simply return the default walk of length $m$, starting in $q_0$. Otherwise, we run $\ennumerate(\uparrow,q_0,m)$ (without making any outputs) and use one integer variable $total$, set to $0$ initially. In this process, each time a call $\ennumerate (a,q,\ell)$ should be made, we first sum up $total$ and the number of walks of length $\ell$ originating in $q$. If this sum is strictly smaller than $i$, then we increase $total$ by the number of walks of length $\ell$ originating in $q$ and skip that call. Otherwise, if the sum is greater or equal to $i$, we make the recursive call. Each time a recursive call is made, we check if $i=total$; if yes, we output the string represented on the stack ${\mathcal S}$, and then stop the process: we have identified the string of rank $i$. The correctness follows immediately, just as in the case of ranking: all is required to identify the string of rank $i$ is a guided root-to-leaf traversal of the tree of recursive calls. The complexity of the algorithm is the same as in the case of the ranking algorithm (by the same arguments).

\smallskip

\noindent \textbf{Theorem \ref{thm:unrank} Unranking.\\ }
% \begin{thm_unranking}
\emph{Given integers $i$ and $m$, and PCA $A$, we can compute the $i^{th}$ string $w$ of length $m$ output in our enumeration algorithm in $O\left(\frac{m\log \sigma}{\tt w}(n\sigma + m (n^{\omega} + m + \sigma))\right)$-time. }
% \end{thm_unranking}

\smallskip

The complexities listed in the results of this section also take into account the time needed to do arithmetical operations on the numbers we work with (in particular, operations involving the variable $total$). To cover this, in the worst case, the final complexity is obtained by multiplying the number of steps done in our algorithms with $\frac{n\log \sigma}{\tt w}$, where ${\tt w}$ is the size of the memory word in our model. However, we note that the overall complexity of both the ranking and the unranking algorithm stays polynomial.

\paragraph{Acknowledgements} The authors thank the reviewers for their helpful comments. Duncan Adamson was supported by the Leverhulme trust via the Leverhulme Centre for Functional Material Design, and by DFG Heisenberg-project number 389613931. Florin Manea was supported by DFG Heisenberg-project number 466789228.

\bibliography{bib_short}
\bibliographystyle{splncs04}

\newpage

\end{document}